\documentclass[prl,superscriptaddress,amsmath,amssymb,twocolumn]{revtex4}

\pdfoutput=1
\usepackage{amsthm}
\usepackage{graphicx}
\usepackage[all]{xy}

\newcommand{\bra}[1]{\langle #1 |}
\newcommand{\ket}[1]{| #1 \rangle}

\newcommand{\ketbra}[2]{\ket{#1}\bra{#2}}
\newcommand{\proj}[1]{\ket{#1}\bra{#1}}
\newcommand{\tr}{{\rm Tr}}
\newcommand{\one}{{\bf 1}}
\newcommand{\re}{{\rm Re}\,}

\newcommand{\prlsection}[1]{{\it{#1}} ---}

\newcommand{\chan}[1]{\mathcal #1}
\newcommand{\cchan}[1]{\widehat {\mathcal #1}}

\newcommand{\id}{{\rm id}}

\theoremstyle{plain}
\newtheorem{proposition}{Proposition}
\newtheorem{theorem}[proposition]{Theorem}

\newtheorem{corollary}[proposition]{Corollary}

\theoremstyle{definition}
\newtheorem{definition}{Definition}

\begin{document}

\title{General conditions for approximate quantum error correction and near-optimal recovery channels}
\author{C\'edric B\'eny}
\affiliation{Centre for Quantum Technologies, National University of Singapore, 3 Science Drive 2, Singapore 117543}
\author{Ognyan Oreshkov}
\affiliation{Grup de F\'{i}sica Te\`{o}rica, Universitat Aut\`{o}noma de Barcelona, 08193 Bellaterra (Barcelona), Spain}
\date{\today}

\begin{abstract}
We derive necessary and sufficient conditions for the approximate correctability of a quantum code, generalizing the Knill-Laflamme conditions for exact error correction. Our measure of success of the recovery operation is the worst-case entanglement fidelity of the overall process. We show that the optimal recovery fidelity can be predicted exactly from a dual optimization problem on the environment causing the noise. We use this result to obtain an easy-to-calculate estimate of the optimal recovery fidelity as well as a way of constructing a class of near-optimal recovery channels that work within twice the minimal error. In addition to standard subspace codes, our results hold for subsystem codes and hybrid quantum-classical codes.
\end{abstract}

\maketitle

\prlsection {Introduction} Given the extreme fragility of quantum coherence, quantum error-correction procedures are believed to be essential for the successful implementation of quantum communication or computation. Exact correctability is characterized in general terms by the Knill-Laflamme (KL) conditions~\cite{knill97} which specify the set of correctable errors for a particular code. However, practically useful codes need not be exactly correctable for any given noise model. In fact, a few exceptional examples show that allowing for a negligible error in the recovery can lead to surprisingly better codes~\cite{leung97,crepeau05}.
This indicates that assuming exact correctability is too strong a restriction. It is therefore of considerable interest to find appropriately weaker error-correction conditions.

In this letter, we generalize the KL conditions to the case of approximate correctability. We view the KL conditions as a statement about the information gathered by the environment causing the noise~\cite{beny09x1} (a fact previously noted in Ref.~\cite{ogawa05}). Thus our analysis makes essential use of the concept of {\em complementary} channel~\cite{devetak05}. Together with tools introduced in Ref.~\cite{kretschmann08}, this provides the basis for our main technical result (Theorem~\ref{main}), which we use to obtain easily computable estimates of the optimal recovery error. We also propose a class of near-optimal recovery channels, which offers a significant simplification to the problem of finding an optimal recovery operation~\cite{yamamoto05,reimpell05,fletcher08,kosut08}.

The analysis of approximate error correction depends on the figure of merit used to compare the states after correction to the input states. In this work we focus on the entanglement fidelity minimized over all input states (also known as worst-case entanglement fidelity). The entanglement fidelity \cite{schumacher96} has been shown to be the pertinent fidelity measure in both quantum communication and computation scenarios since it estimates not only how well the state of the system under correction is preserved but also how its entanglement with auxiliary systems is maintained. Minimization over all inputs is essential if one is interested in guaranteeing a given fidelity when the state to be corrected is not known, as in the case of quantum computing.
In contrast, most previous work has considered input-dependent fidelities~\cite{barnum00,schumacher01,tyson09x1,buscemi08}.
Sufficient conditions for approximate correctability under the worst-case entanglement fidelity were proposed in Ref.~\cite{doddamane09}.
Here we obtain both sufficient and necessary conditions which are a direct generalization of the KL conditions.
Moreover, we prove our result in a very general context; namely for the approximation of any channel, not necessarily the identity map on the code. One advantage of this generality is that our results apply directly to the more general schemes of subsystem, or operator quantum error correction~\cite{kribs05,kribs06,poulin05x1,beny07x1}.
The present results are also strictly stronger than those of Ref.~\cite{kretschmann08x1,beny09x1} which are based on the diamond-norm distance rather than the fidelity.

\prlsection{Background}
The problem of quantum error correction can be formulated as follows: we are given a channel $\chan N$ which can represent either a communication channel or the open dynamics of a physical system which we would like to use as a quantum memory. The goal is to find an encoding operation $\chan E$ and a decoding (or recovery) operation $\chan R$, such that the full operation $\chan R \chan N \chan E$ is equal to the identity map. One usually assumes that the encoding is of the form $\chan E(\rho) = V \rho V^\dagger$ where $V$ is an isometry embedding a small Hilbert space (the code) into the larger physical Hilbert space on which $\chan N$ acts.

Given $\chan N$ and $\chan E$, the KL conditions~\cite{knill97} provide a simple way of testing whether a recovery channel $\chan R$ exists. In addition, these conditions help reasoning about error correction. For instance, one can use them together with the no-cloning theorem to easily demonstrate that it is not possible to encode a qubit in $n$ qubits and faithfully decode it if $n/4$ or more arbitrary qubit errors occur.
The reason we mention this particular example is that it is known to fail dramatically if we allow for an arbitrarily small reconstruction error (provided $n$ is large enough).
Indeed, it was shown in Ref.~\cite{crepeau05} that one can encode quantum information in $n$ qubits undergoing almost $n/2$ arbitrary errors and correct it with vanishing error as $n\rightarrow\infty$.

Here we study what becomes of the KL conditions when we allow for imperfect reconstruction of the code. Additionally, partly because it reveals an important symmetry of the problem,  we also generalize quantum error correction in a different direction.
We seek a ``recovery'' operation $\chan R$ such that $\chan R \chan N \chan E$ is close not necessarily to the identity on the code, but to a fixed arbitrary channel $\chan M$. In particular, this means that our theory applies to subsystem codes~\cite{kribs05,kribs06}, and more generally algebraic codes~\cite{beny07x1} (representing hybrid quantum-classical information), when $\chan M$ projects on an algebra~\cite{beny09x1}.
Note that since we will never separate $\chan N$ from the encoding $\chan E$, we will simply work with a channel ``$\chan N$'' which one can think of as $\chan N \chan E$. It typically maps states on a small (logical) Hilbert space to states on a larger (physical) one.

We will make essential use of the fact that a general quantum operation, or channel $\chan N$, can always be viewed as resulting from a unitary interaction $U$ with an ``environment'' $E$ whose initial state $\ket{\psi}$ is known and which is later discarded (traced out). It does not matter which state $\ket{\psi}$ we use since the difference can be absorbed in the unitary. What matters is the isometry $V$ defined by $V\ket{\phi} := U(\ket{\phi} \otimes \ket{\psi})$ so that
\(
\chan N(\rho) = \tr_E(V \rho V^\dagger).
\)
This isometry $V$ is not unique, but unique up to a further local unitary map on the environment, eventually followed by an embedding into a larger environment. From the isometry $V$, one obtains the channel elements $E_i$ of $\chan N(\rho)=\sum_iE_i\rho E_i^{\dagger}$ simply by writing the partial trace explicitly in terms of a basis $\ket{i}$ of $E$.
If instead of tracing out the environment after the unitary interaction, we trace out the target system $B$, we obtain a channel $\cchan N$ which is said to be {\em complementary} to $\chan N$:
\(
\cchan N(\rho) = \tr_B(V \rho V^\dagger)
\).
It is easy to see that
\begin{equation}
\label{nhatdevelop}
\cchan N(\rho) = \sum\nolimits_{ij} \tr(E_i \rho E_j^\dagger) \ketbra{i}{j}.
\end{equation}
All complementary channels correspond to some choice of the orthonormal family of states $\ket{i}$ in the environment.

\prlsection{Main result}
Let $f(\rho,\sigma) = \tr\sqrt{\sqrt{\rho} {\sigma} \sqrt{\rho}}$ be the fidelity \cite{uhlmann76} between states $\rho$ and $\sigma$. For a given state $\rho$, we introduce the ``entanglement fidelity'' between channels $\chan N$ and $\chan M$,
\begin{equation*}
F_\rho(\chan N, \chan M) = f((\chan N \otimes \id)(\proj{\psi}), (\chan M \otimes \id)(\proj{\psi})),
\end{equation*}
where $\ket{\psi}$ is a purification of $\rho$. When $\chan M = \id$, this quantity reduces to Schumacher's entanglement fidelity of $\chan N$ \cite{schumacher96}. We will compare channels using the worst-case entanglement fidelity
\begin{equation}
F(\chan N, \chan M) = \min_\rho F_\rho(\chan N, \chan M),
\end{equation}
which was studied in Ref.~\cite{gilchrist05}. We remark that $F(\chan N, \chan M)$ relates to $f(\sigma,\rho)$ in the same way that the diamond-norm distance \cite{Kitaev97} relates to the trace distance. Its operational meaning can be deduced from that of $f(\sigma,\rho)$ \cite{Dodd01, Fuchs96}.

\begin{theorem}
\label{main} If $\cchan N$ and $\cchan M$ are channels complementary to $\chan N$ and $\chan M$, respectively, then
\begin{equation}
\label{fidequ} \max_{\chan R} F(\chan R \chan N, \chan M) = \max_{\chan R'} F(\cchan N, \chan R' \cchan M),
\end{equation}
where the maxima are over all quantum channels with the appropriate source and target spaces.
\end{theorem}

\begin{proof}
The proof closely follows arguments used in~\cite{kretschmann08}. Let $V_{\chan N}$ be the isometry for which $\chan N(\rho) = \tr_E (V_{\chan N}\rho V_{\chan N}^\dagger)$ and $\cchan N(\rho) = \tr_B (V_{\chan N}\rho V_{\chan N}^\dagger)$, and $V_{\chan M}$  be the isometry yielding both $\chan M$ and $\cchan M$ in the same way. Note that for a fixed state $\ket{0}$, any channel $\chan R$ can be written as $\chan R(\rho) = \tr_{\widetilde E}(U (\rho \otimes \proj{0}) U^\dagger)$ for some unitary $U$ and appropriate ``environment'' $\widetilde E$. Using this fact and applying Uhlmann's theorem \cite{uhlmann76} which allows us to write the entanglement fidelity in terms of an overlap maximized over unitary operators $U'$, we obtain
\begin{gather}
\label{minimax1}
\max_{\chan R} F(\chan R \chan N, \chan M) 
=\max_{U} \min_{\rho} \max_{U'} |g_\rho(U,U')|,
\end{gather}
where $g_\rho$ can be expressed in terms of a circuit:
\begin{equation*}
g_\rho(U,U') =
\xy0;/r.175pc/:
(-10,12)*{}; (58,12)*{} **\dir{-};    
(-15,7.5)*{    
  \xy
  (-3,6)*{\psi_\rho};
  (0,12)*{}; (0,0)*{} **\crv{(-10,12) & (-10,0)};
  (0,12)*{}; (0,0)*{} **\dir{-};
  (0,12)*{}; (0,0)*{}
  \endxy
};
(-10,4)*{}; (-6,4)*{} **\dir{-};   
(0,-0.5)*{      
  \xy
  (6,6)*{V_{\mathcal N}};
  (0,0)*{}; (12,0)*{} **\dir{-};
  (12,0)*{}; (12,12)*{} **\dir{-};
  (12,12)*{}; (0,12)*{} **\dir{-};
  (0,12)*{}; (0,0)*{} **\dir{-};
  \endxy
};
(6,4)*{}; (14,4)*{} **\dir{-};  (10,6)*{\scriptstyle B}; 
(18,1.5)*{
  \xy
  (4,4)*{U};  
  (0,0)*{}; (8,0)*{} **\dir{-};
  (8,0)*{}; (8,8)*{} **\dir{-};
  (8,8)*{}; (0,8)*{} **\dir{-};
  (0,8)*{}; (0,0)*{} **\dir{-};
  \endxy
};
(11,0)*{}; (14,0)*{} **\dir{-};  
(9,-0.5)*{
  \xy    
  (-1.4,2)*{\scriptstyle 0};
  (0,4)*{}; (0,0)*{} **\crv{(-4,4) & (-4,0)};
  (0,4)*{}; (0,0)*{} **\dir{-};
  (0,4)*{}; (0,0)*{}
  \endxy
};
(6,-4)*{}; (25,-4)*{} **\dir{-}; (16,-7)*{\scriptstyle E};  
(22,0)*{}; (25,0)*{} **\dir{-}; 
(22,4)*{}; (42,4)*{} **\dir{-}; (32,6)*{\scriptstyle B'}; 
(30,-2.5)*{
  \xy    
  (5.2,4)*{{U'}^\dagger};
  (0,0)*{}; (10,0)*{} **\dir{-};
  (10,0)*{}; (10,8)*{} **\dir{-};
  (10,8)*{}; (0,8)*{} **\dir{-};
  (0,8)*{}; (0,0)*{} **\dir{-};
  \endxy
};
(35,0)*{}; (37,0)*{} **\dir{-};  
(35,-4)*{}; (42,-4)*{} **\dir{-}; (38,-7)*{\scriptstyle E'}; 
(39,-0.5)*{
  \xy    
  (1.4,2)*{\scriptstyle 0};
  (0,4)*{}; (0,0)*{} **\crv{(4,4) & (4,0)};
  (0,4)*{}; (0,0)*{} **\dir{-};
  (0,4)*{}; (0,0)*{}
  \endxy
};
(48,-0.5)*{
  \xy    
  (6,6)*{V_{\mathcal M}^\dagger};
  (0,0)*{}; (12,0)*{} **\dir{-};
  (12,0)*{}; (12,12)*{} **\dir{-};
  (12,12)*{}; (0,12)*{} **\dir{-};
  (0,12)*{}; (0,0)*{} **\dir{-};
  \endxy
};
(54,4)*{}; (58,4)*{} **\dir{-};  
(63,7.5)*{
  \xy     
  (3.5,6)*{\psi_\rho};
  (0,12)*{}; (0,0)*{} **\crv{(10,12) & (10,0)};
  (0,12)*{}; (0,0)*{} **\dir{-};
  (0,12)*{}; (0,0)*{}
  \endxy
};
\endxy,
\end{equation*}
where the left half circles represent input states, while the right half circles are states which are scalar multiplied with the corresponding outputs. Hence the picture represents a complex number. The wires labeled $B$ and $B'$ represent the target systems for $\chan N$ and $\chan M$ respectively, and $E$ and $E'$ are the respective ``environments''. The state $\ket{0}$ in the picture is arbitrary, and $\ket{\psi_\rho}$ can be any purification of $\rho$. If we reflect the picture with respect to a vertical axis through the middle, Hermitian conjugating each operator [this amounts to a complex conjugation of $g_{\rho}(U,U')$], and exchange the wire labels $E'$ and $B$, and $E$ and $B'$, we see that we also have $\max_{\chan R'} F( \cchan N, \chan R' \cchan M) =\max_{U'} \min_{\rho} \max_{U} |g_\rho(U,U')|$, where now $U'$ is the unitary defining $\chan R'$ while $U$ comes from Uhlmann's expression for the fidelity. Hence we just have to show that we can exchange the maximizations over $U$ and $U'$ in Eq.~\eqref{minimax1}. First, using the strong concavity of the fidelity, it can be shown that the $\max$ over $U$ in Eq.~\eqref{minimax1} can as well be taken over the convex set of operators $A$ with operator norm $\|A\|\leq 1$. Next, note that \( |g_\rho(A,U')| = |\tr(X_\rho U')| \) for some operator $X_\rho$. We know that $\max_{U'} |\tr(X_\rho U')| = \tr(|X_\rho|)$. Since the optimal value of $\tr(X_\rho U')$ is real, we only need to optimize $\re \tr(X_\rho U')$ which is linear in $\rho$ and $U'$ over the real numbers. In addition, the max over $U'$ can also be taken over operators $A'$ in the unit ball since then $|\tr(X_\rho A')| \le \tr(|X_\rho|)$. We can now apply Shiffman's minimax theorem~\cite{grossinho01} which says that we can exchange the rightmost min and max provided that the function is convex-concave in the two arguments (in this case it is bilinear), and that the variables are optimized over convex sets. Hence, we obtain $\max_{\chan R} F(\chan R \chan N, \chan M) =\max_{A'}\max_{A} \min_{\rho} \re g_\rho(A,A')= \max_{\chan R'} F( \cchan N, \chan R' \cchan M)$, where $\|A'\|, \|A\|\le 1$.
\end{proof}
Note that Eq.~\eqref{fidequ} can be seen as a necessary and sufficient condition for approximate correctability: for a given $\delta\in[0,1]$, there exists a channel $\chan R$ such that $F(\chan R \chan N, \chan M) = 1-\delta$, iff there exists a channel $\chan R'$ such that $F(\cchan N, \chan R' \cchan M)=1-\delta$. We will see that for a large class of problems of interest the existence of $\chan R'$ is much easier to establish than that of $\chan R$.

\prlsection{Knill-Laflamme conditions} Consider the case $\chan M = \id$. An example of a channel complementary to the identity is the trace: $\cchan M = \tr$, whose target is one-dimensional. A channel whose source is one-dimensional outputs a single state. Hence, $\chan R'\cchan M(\rho)=\rho_0$, $\forall \rho$, where $\rho_0$ is a fixed state. Theorem~\ref{main} thus says that $\max_{\chan R} F(\chan R\chan N, \id) = 1$ iff $\cchan N(\rho) = \rho_0 \tr(\rho)$, $\forall \rho$. Explicitly, suppose that the channel $\chan N$ consists of an encoding specified by an isometry $V$ followed by noise with channel elements $E_i$.
In terms of matrix components, the condition that $\cchan N$ be a constant channel with output $\rho_0$ reads
$V^\dagger E_i^\dagger E_j V = \lambda_{ij} \one$, where $\lambda_{ij} = \bra{i} \rho_0 \ket{j}$. One obtains the most familiar form of the KL conditions by expressing these equations using the projector $P = V V^\dagger$ on the code:
\(
P E_i^\dagger E_j P = \lambda_{ij} P.
\)

More generally, if $\chan M = \chan P_{\mathcal A}$ projects on an algebra $\mathcal A$, then we obtain the general correctability conditions for an algebra~\cite{beny07x1}, namely $[A,V^\dagger E_i^\dagger E_j V] = 0$ for all $A \in \mathcal A$. This can be shown by noting that $\cchan P_{\mathcal A} =  \chan P_{\mathcal A'}$, where $\mathcal A'$ is the commutant of $\mathcal A$, i.e. the set of operators commuting with all $A \in \mathcal A$ (see~\cite{beny09x1} for more details.) In particular, when the algebra $\mathcal A$ consists of all operators acting on a subsystem, this yields the conditions for the correctability of a subsystem code~\cite{kribs06}.

Let us show explicitly how Theorem~\ref{main} can be understood as a perturbation of the KL conditions in the case $\chan M = \id$. Since later in our analysis we will use triangle inequalities, it is convenient to measure the error of imperfect recovery by a fidelity-based distance function. We will consider the Bures distance \cite{bures69} based on the entanglement fidelity,
\(
d_\rho(\chan N, \chan M) = \sqrt{1 - {F_\rho(\chan N, \chan M)}}.
\)
Note that
\begin{equation}
d(\chan N, \chan M) := \max_\rho d_\rho(\chan N, \chan M) = \sqrt{1 - {F(\chan N, \chan M)}}
\end{equation}
satisfies the triangle inequality:
\(
d(\chan N, \chan M) = \max_\rho d_\rho(\chan N, \chan M) \le \max_\rho [ d_\rho(\chan N, \chan R) + d_\rho(\chan R, \chan M) ] \le d(\chan N, \chan R) + d(\chan R, \chan M).
\)

\begin{definition}
We will say that a code characterized by the encoding map $\chan E$ is $\varepsilon$-correctable under the noise channel $\chan N$, if there exists a recovery channel $\chan R$ such that $d(\chan R \chan N \chan E, \id)\leq \varepsilon$.
\end{definition}
\begin{corollary}\label{perturbation}
A code defined by the projector $P$ is $\varepsilon$-correctable under a noise channel $\chan N$, if and only if
\begin{equation}
PE_i^{\dagger}E_jP=\lambda_{ij}P+PB_{ij}P, 
\end{equation}
where $\lambda_{ij}$ are the components of a density operator, and
\(
d(\Lambda+\chan B, \Lambda)\leq \varepsilon 
\)
where $\Lambda(\rho)=\sum_{ij}\lambda_{ij}\tr(\rho)|i\rangle\langle j|$ and $(\Lambda+\chan B)(\rho)=\Lambda(\rho)+\sum_{ij}\tr(\rho B_{ij})|i\rangle\langle j|$.
\end{corollary}

\begin{proof}
Let us denote the encoding channel by $\chan E$. It is clear from Theorem~\ref{main} that the code is $\varepsilon$-correctable if and only if there exists a state $\rho_0$ such that $d(\cchan N \chan E, \Lambda) \leq \varepsilon$, where $\Lambda$ is defined as in the statement of the corollary with $\lambda_{ij} = \bra{i}\rho_0 \ket{j}$. Also, from Eq.~\eqref{nhatdevelop} we see that indeed $\cchan N \chan E = \Lambda + \chan B$ since the operators $V^\dagger B_{ij} V$ are defined by $PB_{ij}P = PE_i^{\dagger}E_jP - \lambda_{ij}P$.
\end{proof}

It is not {\it a priori} clear how useful this condition can be since it does not specify how to find an optimal set of coefficients $\lambda_{ij}$. We will now show, in a more general setting, that we can find a whole set of explicit guesses for $\lambda_{ij}$ which are guaranteed to yield a value of $\varepsilon$ which is less than twice the optimal one. Explicitly, this is the case whenever $\lambda_{ij} = \tr(\sigma E_i^\dagger E_j)$ for some state $\sigma$.

\prlsection{Near-optimal correction} We saw that in the exact case (fidelity one), Theorem~\ref{main} yields the necessary and sufficient conditions for all quantum error-correction schemes when $\chan M$ projects on an algebra. Here we want to show that it also yields useful conditions for the approximate version of these schemes. The problem is that in general it may not be easier to compute $\max_{\chan R'} F(\cchan N, \chan R' \cchan M)$ than $\max_{\chan R} F(\chan R \chan N, \chan M)$. However, we will show that when $\cchan M$ is a projection, i.e. it satisfies $\cchan M^2 = \cchan M$ (which is the case for error correction), we can guess a whole class of channels $\widetilde{\chan R}'$ for which $F(\cchan N, \widetilde{\chan R}' \cchan M)$ yields a good approximation to the optimal worst-case fidelity $\max_{\chan R} F(\chan R \chan N, \chan M)$. Moreover, we can build the corresponding near-optimal recovery channels $\widetilde{\chan R}$.

\begin{corollary}
\label{useful}
Suppose that $\cchan M^2 = \cchan M$. Then
\begin{equation}
\frac{1}{2}d(\cchan N, \cchan N \cchan M) \le \min_{\chan R} d(\chan R \chan N, \chan M) \le d(\cchan N, \cchan N \cchan M).
\end{equation}
\end{corollary}
\proof
First note that Theorem~\ref{main} expressed in terms of $d$ reads $\varepsilon_0 := \min_{\chan R} d(\chan R \chan N, \chan M) = \min_{\chan R'} d(\cchan N, \chan R' \cchan M)$. The rightmost inequality follows from picking the nonoptimal $\chan R' = \cchan N$. For the leftmost inequality, suppose that $\chan R'_0$ minimizes $d(\cchan N, \chan R' \cchan M)$. Then using the triangle inequality,
\(
d(\cchan N, \cchan N \cchan M) \le d(\cchan N, \chan R'_0 \cchan M) + d(\chan R'_0 \cchan M, \cchan N \cchan M).
\)
We know that the first term is equal to $\varepsilon_0$ since $\chan R'_0$ is optimal. For the second term, note that
\(
 d(\chan R'_0 \cchan M, \cchan N \cchan M) = d(\chan R'_0 \cchan M^2, \cchan N \cchan M) \le d(\chan R'_0\cchan M, \cchan N)  = \varepsilon_0.
\)
Hence, $d(\cchan N, \cchan N \cchan M) \le 2\varepsilon_0$.
\qed

Note that computing $d(\cchan N, \cchan N \cchan M)$ requires a convex maximization over inputs only \cite{gilchrist05}, which is a significant simplification over the minimax $\min_{\chan R} d(\chan R\chan N,\chan M)$.

\prlsection{Near-optimal recovery channel}
Let us show how we can construct a recovery channel $\widetilde{\chan R}$ which performs as well as guaranteed by our bounds, i.e.
\begin{equation}
\label{goodchan}
d(\widetilde{\chan R} \chan N, \chan M) \le d(\cchan N, \cchan N\cchan M).
\end{equation}
We first need to find a saddle point $(\rho_0,A_0)$ of $\re g_{\rho}(A,U')$, where $U'$ yields $\cchan N$ through $\cchan N(\rho) = \tr_2[U'(\rho \otimes \proj{0}) (U')^\dagger]$. This implies that $F(\cchan N, \cchan N\cchan M) = \re g_{\rho_0}(A_0,U')$.
For instance, in the case $\chan M=\id$, one may first find a $\rho_0$ that minimizes $\tr(|X_{\rho_0}|)$, which is a convex optimization problem~\cite{gilchrist05}. If $\rho_0$ is full rank and unique, then $A_0$ can be chosen to be any unitary obtained from the polar decomposition of $X_{\rho_0}$.
Generally, the saddle-point operator $A_0$ yields a channel
\(
\chan S(\rho) := \tr_2(A_0(\rho \otimes \proj{0}) A_0^\dagger)
\)
which may be trace-decreasing but can always be completed to a trace-preserving channel $\widetilde{\chan R}(\rho) = \chan S(\rho) + \tr(\rho - \chan S(\rho)) \tau$ for any fixed state $\tau$.
$\widetilde{\chan R}$ then satisfies Eq.~\eqref{goodchan}. Indeed,
$F(\cchan N, \cchan N\cchan M) = \min_\rho \re g_{\rho}(A_0,U')
\le \max_{A',\|A'\|\le 1} \min_\rho \re g_{\rho}(A_0,A')
= \min_\rho \max_{U'} \re g_{\rho}(A_0,U')
\le F(\widetilde{\chan R} \chan N, \chan M).$

\begin{figure}
\includegraphics[width=0.35\textwidth]{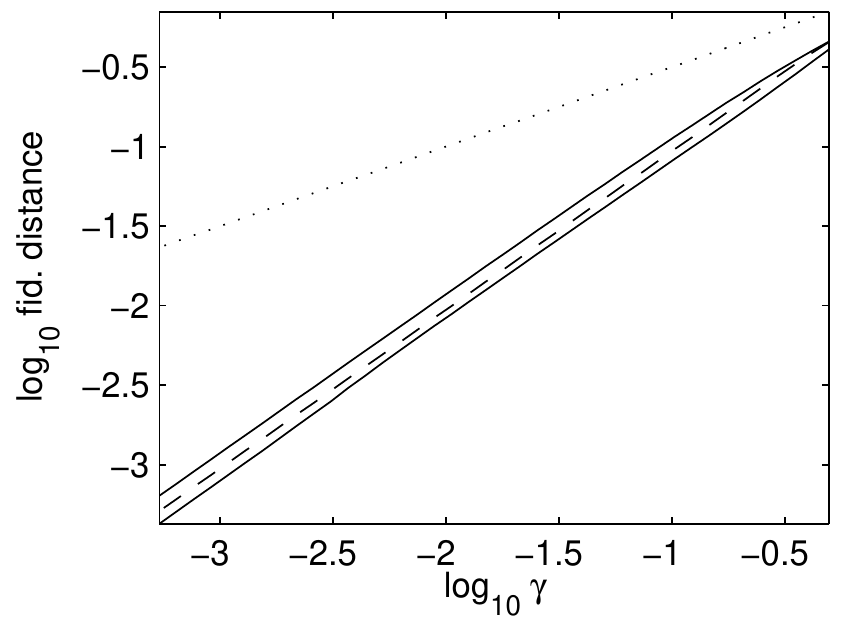}
\caption{Example based on the encoding $\chan E$ and channel $\chan N_{\gamma}$ of Ref.~\cite{leung97}.
Dotted line: fidelity distance without correction.
Upper solid line: distance achieved by the recovery channel described in Ref.~\cite{leung97}.
Dashed line: near-optimal estimate $\delta(\chan N_\gamma \chan E)$.
Lower solid line: distance achieved by our near-optimal recovery channel.}
\label{zefig}
\end{figure}

\prlsection{Example}
In the standard case $\chan M = \id$, the estimate
\begin{equation}
\delta(\chan N) = d(\cchan N,\cchan N \cchan M)
\end{equation}
from Corollary~\ref{useful} is given explicitly in terms of the fidelity by
\begin{equation}
F(\cchan N,\cchan N \cchan M) = \min_\rho \tr\sqrt{\sum\nolimits_{ij} E_i \rho^2 E_j^\dagger \, \tr(E_j \sigma E_i^\dagger)},
\end{equation}
where $\sigma$ is an arbitrary state resulting from the freedom in choosing the complementary channel $\cchan M$.
As an example, we applied our results to the code proposed in Ref.~\cite{leung97}, for which we numerically computed $\delta(\chan N \chan E)$, where $\chan E$ is the encoding, and $\chan N_\gamma$ is the noise model with strength $\gamma$ considered there (FIG.~\ref{zefig}).
We see that our estimate predicts that this code is good in terms of the worst-case entanglement fidelity, in the sense that it yields a fidelity distance of order $\gamma$ instead of the uncorrected order $\sqrt \gamma$. The necessity of our criterion shows in particular that one cannot obtain a better asymptotic behavior with this code.

\prlsection{Fixed-state fidelity}
We finally note that our method can also be applied to the problem of error correction on a fixed input state (a scenario studied, e.g., in Ref.~\cite{schumacher01,barnum00}). Indeed, for any state $\rho$ we also have $\max_{\chan R} F_{\rho}(\chan R \chan N, \chan M) = \max_{\chan R'} F_{\rho}(\cchan N, \chan R' \cchan M)$. The argument is simpler as no minimization over $\rho$ is involved.
In the case $\chan M = \id$, and using a reasoning very similar to the one we used for the worst-case fidelity, we can obtain the estimate
\begin{equation}
\frac{1}{2}d_{\rho}(\cchan N, \chan S) \le \min_{\chan R} d_{\rho}(\chan R \chan N, \id) \le d_{\rho}(\cchan N,  \chan S)
\end{equation}
where $\chan S(\sigma) := \cchan N(\rho) \tr(\sigma)$.
In addition, the corresponding near-optimal recovery channel can be built using the same method as for the worst-case fidelity with the simplification that no search of a saddle point is involved.

\prlsection{Conclusion}
In summary, we have generalized the Knill-Laflamme conditions to the case of approximate correctability, including standard codes, subsystem codes, and hybrid quantum-classical codes.
We obtained easy-to-calculate estimates of the optimal recovery error and proposed a class of near-optimal recovery channels for the worst-case entanglement fidelity that work within twice of the optimal error. These results provide a framework for studying error correction under general noise models and allow for a significant simplification to the task of finding optimal error-correction procedures, thus offering a promising tool to efficiently address the problem of decoherence control in realistic scenarios.

\prlsection{Acknowledgements} OO was supported by Spanish MICINN (Consolider-Ingenio QOIT). Part of this work was done during the QI workshop at the Benasque Center for Science, Benasque, Spain, and during the Fields Institute Thematic Program on Mathematics in Quantum Information. The Centre for Quantum Technologies is funded by the Singapore Ministry of Education and the National Research Foundation as part of the Research Centres of Excellence programme.

\bibliography{approx_qec2}

\begin{thebibliography}{29}
\expandafter\ifx\csname natexlab\endcsname\relax\def\natexlab#1{#1}\fi
\expandafter\ifx\csname bibnamefont\endcsname\relax
  \def\bibnamefont#1{#1}\fi
\expandafter\ifx\csname bibfnamefont\endcsname\relax
  \def\bibfnamefont#1{#1}\fi
\expandafter\ifx\csname citenamefont\endcsname\relax
  \def\citenamefont#1{#1}\fi
\expandafter\ifx\csname url\endcsname\relax
  \def\url#1{\texttt{#1}}\fi
\expandafter\ifx\csname urlprefix\endcsname\relax\def\urlprefix{URL }\fi
\providecommand{\bibinfo}[2]{#2}
\providecommand{\eprint}[2][]{\url{#2}}

\bibitem[{\citenamefont{Knill and Laflamme}(1997)}]{knill97}
\bibinfo{author}{\bibfnamefont{E.}~\bibnamefont{Knill}} \bibnamefont{and}
  \bibinfo{author}{\bibfnamefont{R.}~\bibnamefont{Laflamme}},
  \bibinfo{journal}{Phys. Rev. A} \textbf{\bibinfo{volume}{55}},
  \bibinfo{pages}{900} (\bibinfo{year}{1997}).

\bibitem[{\citenamefont{Leung et~al.}(1997)\citenamefont{Leung, Nielsen,
  Chuang, and Yamamoto}}]{leung97}
\bibinfo{author}{\bibfnamefont{D.~W.} \bibnamefont{Leung}},
  \bibinfo{author}{\bibfnamefont{M.~A.} \bibnamefont{Nielsen}},
  \bibinfo{author}{\bibfnamefont{I.~L.} \bibnamefont{Chuang}},
  \bibnamefont{and} \bibinfo{author}{\bibfnamefont{Y.}~\bibnamefont{Yamamoto}},
  \bibinfo{journal}{Phys. Rev. A} \textbf{\bibinfo{volume}{56}},
  \bibinfo{pages}{2567} (\bibinfo{year}{1997}), \eprint{quant-ph/9704002}.

\bibitem[{\citenamefont{Cr\'epeau et~al.}(2005)\citenamefont{Cr\'epeau,
  Gottesman, and Smith}}]{crepeau05}
\bibinfo{author}{\bibfnamefont{C.}~\bibnamefont{Cr\'epeau}},
  \bibinfo{author}{\bibfnamefont{D.}~\bibnamefont{Gottesman}},
  \bibnamefont{and} \bibinfo{author}{\bibfnamefont{A.}~\bibnamefont{Smith}},
  \bibinfo{journal}{Lecture Notes in Computer Science}
  \textbf{\bibinfo{volume}{3494}}, \bibinfo{pages}{285} (\bibinfo{year}{2005}),
  \eprint{quant-ph/0503139}.

\bibitem[{\citenamefont{B\'eny}(2009)}]{beny09x1}
\bibinfo{author}{\bibfnamefont{C.}~\bibnamefont{B\'eny}},
  \bibinfo{journal}{Lecture Notes in Computer Science}
  \textbf{\bibinfo{volume}{5906}}, \bibinfo{pages}{6675}
  (\bibinfo{year}{2009}), \eprint{arXiv:0907.4207}.

\bibitem[{\citenamefont{Ogawa}(2005)}]{ogawa05}
\bibinfo{author}{\bibfnamefont{T.}~\bibnamefont{Ogawa}},
  \emph{\bibinfo{title}{Perfect quantum error-correcting condition revisited}}
  (\bibinfo{year}{2005}), \eprint{arXiv:quant-ph/0505167}.

\bibitem[{\citenamefont{Devetak and Shor}(2005)}]{devetak05}
\bibinfo{author}{\bibfnamefont{I.}~\bibnamefont{Devetak}} \bibnamefont{and}
  \bibinfo{author}{\bibfnamefont{P.~W.} \bibnamefont{Shor}},
  \bibinfo{journal}{Communications in Mathematical Physics}
  \textbf{\bibinfo{volume}{256}}, \bibinfo{pages}{287} (\bibinfo{year}{2005}).

\bibitem[{\citenamefont{Kretschmann
  et~al.}(2008{\natexlab{a}})\citenamefont{Kretschmann, Schlingemann, and
  Werner}}]{kretschmann08}
\bibinfo{author}{\bibfnamefont{D.}~\bibnamefont{Kretschmann}},
  \bibinfo{author}{\bibfnamefont{D.}~\bibnamefont{Schlingemann}},
  \bibnamefont{and} \bibinfo{author}{\bibfnamefont{R.}~\bibnamefont{Werner}},
  \bibinfo{journal}{IEEE Transactions on Information Theory}
  \textbf{\bibinfo{volume}{54}}, \bibinfo{pages}{1708}
  (\bibinfo{year}{2008}{\natexlab{a}}), \eprint{quant-ph/0605009}.

\bibitem[{\citenamefont{Fletcher et~al.}(2008)\citenamefont{Fletcher, Shor, and
  Win}}]{fletcher08}
\bibinfo{author}{\bibfnamefont{A.~S.} \bibnamefont{Fletcher}},
  \bibinfo{author}{\bibfnamefont{P.~W.} \bibnamefont{Shor}}, \bibnamefont{and}
  \bibinfo{author}{\bibfnamefont{M.~Z.} \bibnamefont{Win}},
  \bibinfo{journal}{Phys. Rev. A} \textbf{\bibinfo{volume}{77}},
  \bibinfo{pages}{012320} (\bibinfo{year}{2008}).

\bibitem[{\citenamefont{Kosut et~al.}(2008)\citenamefont{Kosut, Shabani, and
  Lidar}}]{kosut08}
\bibinfo{author}{\bibfnamefont{R.~L.} \bibnamefont{Kosut}},
  \bibinfo{author}{\bibfnamefont{A.}~\bibnamefont{Shabani}}, \bibnamefont{and}
  \bibinfo{author}{\bibfnamefont{D.~A.} \bibnamefont{Lidar}},
  \bibinfo{journal}{Phys. Rev. Lett.} \textbf{\bibinfo{volume}{100}},
  \bibinfo{pages}{020502} (\bibinfo{year}{2008}).

\bibitem[{\citenamefont{Yamamoto et~al.}(2005)\citenamefont{Yamamoto, Hara, and
  Tsumura}}]{yamamoto05}
\bibinfo{author}{\bibfnamefont{N.}~\bibnamefont{Yamamoto}},
  \bibinfo{author}{\bibfnamefont{S.}~\bibnamefont{Hara}}, \bibnamefont{and}
  \bibinfo{author}{\bibfnamefont{K.}~\bibnamefont{Tsumura}},
  \bibinfo{journal}{Phys. Rev. A} \textbf{\bibinfo{volume}{71}},
  \bibinfo{pages}{022322} (\bibinfo{year}{2005}).

\bibitem[{\citenamefont{Reimpell and Werner}(2005)}]{reimpell05}
\bibinfo{author}{\bibfnamefont{M.}~\bibnamefont{Reimpell}} \bibnamefont{and}
  \bibinfo{author}{\bibfnamefont{R.~F.} \bibnamefont{Werner}},
  \bibinfo{journal}{Phys. Rev. Lett.} \textbf{\bibinfo{volume}{94}},
  \bibinfo{pages}{080501} (\bibinfo{year}{2005}).

\bibitem[{\citenamefont{Schumacher}(1996)}]{schumacher96}
\bibinfo{author}{\bibfnamefont{B.~W.} \bibnamefont{Schumacher}},
  \bibinfo{journal}{Phys. Rev. A} \textbf{\bibinfo{volume}{54}},
  \bibinfo{pages}{2614} (\bibinfo{year}{1996}).

\bibitem[{\citenamefont{Barnum and Knill}(2002)}]{barnum00}
\bibinfo{author}{\bibfnamefont{H.}~\bibnamefont{Barnum}} \bibnamefont{and}
  \bibinfo{author}{\bibfnamefont{E.}~\bibnamefont{Knill}}, \bibinfo{journal}{J.
  Math. Phys.} \textbf{\bibinfo{volume}{43}}, \bibinfo{pages}{2097}
  (\bibinfo{year}{2002}), \eprint{quant-ph/0004088}.

\bibitem[{\citenamefont{Schumacher and Westmoreland}(2002)}]{schumacher01}
\bibinfo{author}{\bibfnamefont{B.}~\bibnamefont{Schumacher}} \bibnamefont{and}
  \bibinfo{author}{\bibfnamefont{M.~D.} \bibnamefont{Westmoreland}},
  \bibinfo{journal}{Quantum Information Processing}
  \textbf{\bibinfo{volume}{1}}, \bibinfo{pages}{5} (\bibinfo{year}{2002}).

\bibitem[{\citenamefont{Tyson}(2009)}]{tyson09x1}
\bibinfo{author}{\bibfnamefont{J.}~\bibnamefont{Tyson}},
  \emph{\bibinfo{title}{Two-sided bounds on minimum-error quantum measurement,
  on the reversibility of quantum dynamics, and on the maximum overlap problem
  using abstract jezek-rehacek-fiurasek-hradil iterates}}
  (\bibinfo{year}{2009}), \eprint{arXiv:0907.3386}.

\bibitem[{\citenamefont{Buscemi}(2008)}]{buscemi08}
\bibinfo{author}{\bibfnamefont{F.}~\bibnamefont{Buscemi}},
  \bibinfo{journal}{Phys. Rev. A} \textbf{\bibinfo{volume}{77}},
  \bibinfo{pages}{012309} (\bibinfo{year}{2008}).

\bibitem[{\citenamefont{Mandayam and Poulin}(2007)}]{doddamane09}
\bibinfo{author}{\bibfnamefont{P.}~\bibnamefont{Mandayam}} \bibnamefont{and}
  \bibinfo{author}{\bibfnamefont{D.}~\bibnamefont{Poulin}},
  \bibinfo{journal}{First International Conference on Quantum Error Correction}
   (\bibinfo{year}{2007}).

\bibitem[{\citenamefont{Kribs et~al.}(2005)\citenamefont{Kribs, Laflamme, and
  Poulin}}]{kribs05}
\bibinfo{author}{\bibfnamefont{D.}~\bibnamefont{Kribs}},
  \bibinfo{author}{\bibfnamefont{R.}~\bibnamefont{Laflamme}}, \bibnamefont{and}
  \bibinfo{author}{\bibfnamefont{D.}~\bibnamefont{Poulin}},
  \bibinfo{journal}{Physical Review Letters} \textbf{\bibinfo{volume}{94}},
  \bibinfo{pages}{180501} (\bibinfo{year}{2005}).

\bibitem[{\citenamefont{Kribs et~al.}(2006)\citenamefont{Kribs, Laflamme,
  Poulin, and Lesosky}}]{kribs06}
\bibinfo{author}{\bibfnamefont{D.~W.} \bibnamefont{Kribs}},
  \bibinfo{author}{\bibfnamefont{R.}~\bibnamefont{Laflamme}},
  \bibinfo{author}{\bibfnamefont{D.}~\bibnamefont{Poulin}}, \bibnamefont{and}
  \bibinfo{author}{\bibfnamefont{M.}~\bibnamefont{Lesosky}},
  \bibinfo{journal}{Quantum Information and Computation}
  \textbf{\bibinfo{volume}{6}}, \bibinfo{pages}{382} (\bibinfo{year}{2006}),
  \eprint{arXiv:quant-ph/0504189}.

\bibitem[{\citenamefont{Poulin}(2005)}]{poulin05x1}
\bibinfo{author}{\bibfnamefont{D.}~\bibnamefont{Poulin}},
  \bibinfo{journal}{Phys. Rev. Lett.} \textbf{\bibinfo{volume}{95}},
  \bibinfo{pages}{230504} (\bibinfo{year}{2005}).

\bibitem[{\citenamefont{Beny et~al.}(2007)\citenamefont{Beny, Kempf, and
  Kribs}}]{beny07x1}
\bibinfo{author}{\bibfnamefont{C.}~\bibnamefont{Beny}},
  \bibinfo{author}{\bibfnamefont{A.}~\bibnamefont{Kempf}}, \bibnamefont{and}
  \bibinfo{author}{\bibfnamefont{D.~W.} \bibnamefont{Kribs}},
  \bibinfo{journal}{Phys. Rev. Lett.} \textbf{\bibinfo{volume}{98}},
  \bibinfo{pages}{100502} (\bibinfo{year}{2007}).

\bibitem[{\citenamefont{Kretschmann
  et~al.}(2008{\natexlab{b}})\citenamefont{Kretschmann, Kribs, and
  Spekkens}}]{kretschmann08x1}
\bibinfo{author}{\bibfnamefont{D.}~\bibnamefont{Kretschmann}},
  \bibinfo{author}{\bibfnamefont{D.~W.} \bibnamefont{Kribs}}, \bibnamefont{and}
  \bibinfo{author}{\bibfnamefont{R.~W.} \bibnamefont{Spekkens}},
  \bibinfo{journal}{Physical Review A} \textbf{\bibinfo{volume}{78}},
  \bibinfo{pages}{032330} (\bibinfo{year}{2008}{\natexlab{b}}).

\bibitem[{\citenamefont{Uhlmann}(1976)}]{uhlmann76}
\bibinfo{author}{\bibfnamefont{A.}~\bibnamefont{Uhlmann}},
  \bibinfo{journal}{Rep. Math. Phys.} \textbf{\bibinfo{volume}{9}},
  \bibinfo{pages}{273} (\bibinfo{year}{1976}).

\bibitem[{\citenamefont{Gilchrist et~al.}(2005)\citenamefont{Gilchrist,
  Langford, and Nielsen}}]{gilchrist05}
\bibinfo{author}{\bibfnamefont{A.}~\bibnamefont{Gilchrist}},
  \bibinfo{author}{\bibfnamefont{N.~K.} \bibnamefont{Langford}},
  \bibnamefont{and} \bibinfo{author}{\bibfnamefont{M.~A.}
  \bibnamefont{Nielsen}}, \bibinfo{journal}{Phys. Rev. A}
  \textbf{\bibinfo{volume}{71}}, \bibinfo{pages}{062310}
  (\bibinfo{year}{2005}).

\bibitem[{\citenamefont{Kitaev}(1997)}]{Kitaev97}
\bibinfo{author}{\bibfnamefont{A.~Y.} \bibnamefont{Kitaev}},
  \bibinfo{journal}{Russ. Math. Surv.} \textbf{\bibinfo{volume}{52}},
  \bibinfo{pages}{1191} (\bibinfo{year}{1997}).

\bibitem[{\citenamefont{Dodd and Nielsen}(2001)}]{Dodd01}
\bibinfo{author}{\bibfnamefont{J.~L.} \bibnamefont{Dodd}} \bibnamefont{and}
  \bibinfo{author}{\bibfnamefont{M.~A.} \bibnamefont{Nielsen}},
  \emph{\bibinfo{title}{A simple operational interpretation of the fidelity}}
  (\bibinfo{year}{2001}), \eprint{arXiv:quant-ph/0111053}.

\bibitem[{\citenamefont{Fuchs}(1996)}]{Fuchs96}
\bibinfo{author}{\bibfnamefont{C.~A.} \bibnamefont{Fuchs}},
  \emph{\bibinfo{title}{Distinguishability and accessible information in
  quantum theory}} (\bibinfo{year}{1996}), \eprint{arXiv:quant-ph/9601020}.

\bibitem[{\citenamefont{do~Rosario~Grossinho and Tersian}(2001)}]{grossinho01}
\bibinfo{author}{\bibfnamefont{M.}~\bibnamefont{do~Rosario~Grossinho}}
  \bibnamefont{and} \bibinfo{author}{\bibfnamefont{S.~A.}
  \bibnamefont{Tersian}}, \emph{\bibinfo{title}{An introduction to minimax
  theorems and their applications to differential equations}}
  (\bibinfo{publisher}{Kluwer}, \bibinfo{year}{2001}).

\bibitem[{\citenamefont{Bures}(1969)}]{bures69}
\bibinfo{author}{\bibfnamefont{D.}~\bibnamefont{Bures}},
  \bibinfo{journal}{Trans. Am. Math. Soc.} \textbf{\bibinfo{volume}{135}},
  \bibinfo{pages}{199} (\bibinfo{year}{1969}).

\end{thebibliography}

\end{document}